%% file: main_final.tex
\PassOptionsToPackage{dvipsnames}{xcolor}
\documentclass[journal, twocolumn, a4paper]{IEEEtran}
\usepackage[utf8]{inputenc}
\usepackage{cite}
\usepackage{amsfonts, amsmath, amssymb, amsthm}
\usepackage[utf8]{inputenc}
\usepackage{array, makecell}
\usepackage{mathtools,graphicx}
\usepackage{tikz, color, soul}
\usepackage{pgfplots}
\usepackage{bbold}
\usepackage{enumitem}
\usepackage{hyperref,cleveref}
\usepackage{lipsum}
\usepackage{filecontents}
\newtheorem{theorem}{Theorem}
\newtheorem{proposition}{Proposition}
\newtheorem{lemma}[theorem]{Lemma}

\newtheorem{corollary}{Corollary}

\newtheorem{remark}{Remark}
\theoremstyle{definition}
\newtheorem{definition}{Definition}

\newcommand{\C}{\mathcal{C}}

\renewcommand{\v}{\boldsymbol{v}}

\newcommand{\wt}{\textrm{wt}}

\newcommand{\x}{\boldsymbol{x}}
\renewcommand{\L}{\mathcal{L}}
\newcommand{\y}{\boldsymbol{y}}
\newcommand{\z}{\boldsymbol{z}}
\newcommand{\0}{\boldsymbol{0}}
\newcommand{\1}{\boldsymbol{1}}

\renewcommand{\epsilon}{\varepsilon}
\newcommand{\E}{\mathbf{E}}
\newcommand{\s}{\mathrm{supp}}

\definecolor{darkgray}{RGB}{64,64,64}
\definecolor{litegray}{RGB}{192,192,192}
\tikzstyle{block}=[draw, rectangle, minimum height=1cm, text width=1.3cm, text centered, draw=darkgray, font=\small]
\tikzstyle{block_medium}=[draw, rectangle, minimum height=1.5cm, text width=2cm, text centered, draw=darkgray, font=\small]
\tikzstyle{block_large}=[draw, rectangle, minimum height=2cm, text width=2cm, text centered, draw=darkgray, font=\small]
\tikzstyle{line} = [draw, -latex]

\newcommand\myshade{85}
\colorlet{mylinkcolor}{violet}
\colorlet{mycitecolor}{YellowOrange}
\colorlet{myurlcolor}{Aquamarine}

\hypersetup{
	linkcolor  = mylinkcolor!\myshade!black,
	citecolor  = mycitecolor!\myshade!black,
	urlcolor   = myurlcolor!\myshade!black,
	colorlinks = true,
}
\title{Two-stage coding over the Z-channel}
\author{
 \IEEEauthorblockN{Lebedev Alexey, Lebedev Vladimir, and Polyanskii Nikita}
\thanks{A. Lebedev and V. Lebedev were supported by the Russian Foundation for Basic Research (RFBR) under Grant No. 19-01-00364 and by the RFBR and the Japan Society for the Promotion of Science (JSPS) under Grant No. 20-51-50007. N. Polyanskii's research was conducted in part during June 2020 - December 2021 with the Technical University of Munich and the Skolkovo Institute of Science and Technology. His work was supported by the German Research Foundation (Deutsche Forschungsgemeinschaft, DFG) under Grant No. WA3907/1-1 and the Russian Foundation for Basic Research (RFBR) under Grant No.~\mbox{20-01-00559}.}
\thanks{A.~Lebedev and V.~Lebedev are with the  Institute for Information Transmission Problems, Russian Academy of Sciences, Russia.  N.~Polyanskii is with the IOTA Foundation, Germany.}
}
\newif\ifcomment
\commenttrue

\pgfplotsset{compat = 1.15}
\begin{document}

\maketitle
\begin{abstract}
    In this paper, we discuss two-stage encoding algorithms capable of correcting a fraction of asymmetric errors. Suppose that the encoder transmits $n$ binary symbols $(x_1,\ldots,x_n)$ one-by-one over the Z-channel, in which a 1 is received only if a 1 is transmitted. At some designated moment, say $n_1$, the encoder uses noiseless feedback  and adjusts further encoding strategy based on the partial output of the channel $(y_1,\ldots,y_{n_1})$. The goal is to transmit error-free as much information as possible under the assumption that the total number of errors inflicted by the Z-channel is limited by $\tau n$, $0<\tau<1$.  We propose an encoding strategy that uses a list-decodable code at the first stage and a high-error low-rate code at the second stage. This strategy and our converse result yield that there is a sharp transition at $\tau=\max\limits_{0<w<1}\frac{w + w^3}{1+4w^3}\approx 0.44$ from positive rate to zero rate for two-stage encoding strategies. As side results,  we derive bounds on the size of list-decodable codes for the Z-channel and prove that for a fraction $1/4+\epsilon$ of asymmetric errors, an error-correcting code contains at most  $O(\epsilon^{-3/2})$ codewords.
\end{abstract}
\section{Introduction}
 The Z-channel 
 is of asymmetric nature because it permits an error $1\to 0$, whereas it prohibits an error $0\to 1$. The problem of finding encoding strategies for the Z-channel with one use of the noiseless feedback is addressed in this paper. We consider the combinatorial setting in which we limit the maximal number of errors inflicted by the channel by $\tau n$, where $\tau$ is a real number and $n$ denotes the number of channel uses. 	We emphasize that our combinatorial model is different from the probabilistic model, in which a transmitted symbol $1$ is flipped with probability $p$ and a symbol $0$ is always received without error. Recall that the capacity of the Z-channel is $C_{Z}=\log_2(1+p^{p/(1-p)}(1-p))$. We refer the reader to the paper~\cite{tallini2008feedback} dealing with the probabilistic setting,  where feedback encoding schemes that achieve the Z-channel capacity are presented.
\subsection{Related work}
 We briefly review the combinatorial coding theory literature relevant to our research. Without feedback, codes correcting asymmetric errors have been discussed in numerous papers~\cite{borden1983low,kim1959single,varshamov1965theory,klove1981upper,klove1981error,fu2003new,bose1993asymmetric,zhang2019construction}. In particular, it is known~\cite{bassalygo1965new,borden1983low} that the asymptotic rate of codes correcting a fraction of asymmetric errors is equal to the asymptotic rate of codes correcting the same fraction of symmetric errors. The Plotkin bound~\cite{plotkin1960binary} implies that the cardinality of codes  correcting a fraction $(1/4+\epsilon)$ of symmetric errors is bounded by $1+1/(4\epsilon)$ and, thus, the asymptotic rate is zero. Therefore, the asymptotic rate of codes correcting a fraction $(1/4+\epsilon)$ of asymmetric errors is also zero. However, there remains a question of whether it is possible to construct codes of length $n$ with an arbitrary large size capable of correcting  $n(1/4+\epsilon)$ asymmetric errors. A Plotkin-type bound based on linear programming arguments was derived in~\cite{borden1983low}. Moreover, it was claimed that  ``there are arbitrarily large codes which can correct $1/3$ asymmetric errors per code letter''. 
As a side result, in Section~\ref{sec::high error low rate codes}, we disprove this statement.

Under the guise of a half-lie game, coding over the Z-channel with the noiseless feedback has been discussed in \cite{rivest1980coping,cicalese2000optimal,dumitriu2004halfliar,spencer2003halflie}. However, in all these papers, only a constant number of errors and fully adaptive strategies were assumed. Later, in~\cite{dumitriu2005two} the authors have shown that for a constant number of errors, it is sufficient to use the feedback only once to transmit asymptotically the same amount of messages. Feedback codes correcting a fraction $\tau$ of errors in the Z-channel have been discussed in~\cite{dyachkov11upper,deppe2020coding,deppe2020bounds}. In particular, it was shown~\cite{deppe2020coding} that for any $\tau<1$, the maximal asymptotic rate is positive. 

\subsection{Problem statement}
In this paper, we discuss the asymptotic rate of codes correcting a fraction $\tau$ of asymmetric errors if one use of the feedback is allowed. This setting can be seen as a compromise between fully adaptive encoding strategies and classic error-correcting codes. As mentioned above this type of problem has been discussed before only for a constant number of errors and the methods developed in~\cite{dumitriu2005two} can not be used here. Let us describe the model in more details.

Let $t$ denote the total number of errors and $n$ be the total number of channel uses. We fix some integer $n_1$ such that $1< n_1 < n$. For any message $m\in[M]$, Alice wishes to encode it to a string $\x=(x_1,\ldots,x_n)$ such that after transmitting this string through the Z-channel, Bob would be able to correctly decode the message. The output string $\y=(y_1,\ldots,y_n)$ is controlled by an adversary Calvin who can make errors. At the $i$th moment, Alice generates a binary symbol $x_i$, and Calvin takes $x_i$ and outputs $y_i$, where
$$
y_i
\in
\begin{cases}
\{0\},\quad &\text{if } x_i=0,\\
\{0,1\},\quad &\text{if } x_i=1.
\end{cases}
$$
  The process of encoding a string $\x$ consists of two stages.  At the $(n_1+1)$th moment, Alice  adapts the further encoding strategy for the message $m$ based on the string $\y^{n_1}:=(y_1,\ldots,y_{n_1})$. In other words, $x_i=x_i(m)$ for $i\le n_1$ and $x_{i}=x_{i}(m,\y^{n_1})$ for $i>n_1$. We require the total number of errors that Calvin can produce to be at most $t$. This communication scheme is depicted in Figure~\ref{fig::communication channel}.
 Let $M_{Z}^{(2)}(n,t)$ be the maximum number of messages Alice can transmit to Bob under conditions imposed by this model. We remark that the optimal moment for using the feedback, $n_1$, depends on $t$ and $n$. However, Alice and Bob can agree on this parameter beforehand to maximize the total number of messages. Define the \textit{maximal asymptotic rate} of two-stage error-correcting  codes for the Z-channel to be
 $$
 R_{Z}^{(2)}(\tau):=\limsup_{n\to\infty} \frac{\log M_{Z}^{(2)}(n,\lfloor \tau n \rfloor)}{n}.
 $$
 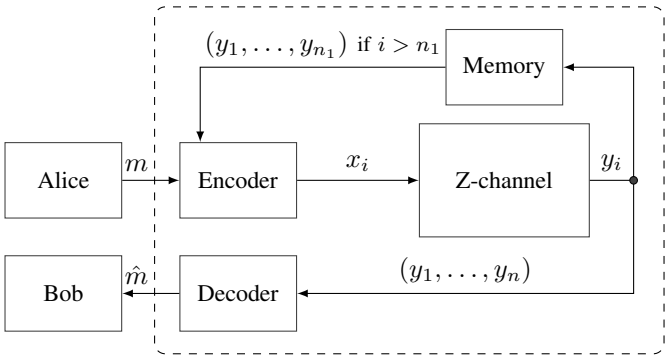
\begin{figure}[t]
		\centering
		\begin{tikzpicture}
		\node[block_medium] (c) at (-0.2,0) {Z-channel};
		\node[block] (d) at (-0.2, 1.5) {Memory};
		\node[block] (e1) at (-3.7,0) {Encoder};
		\node[block] (u1) at (-6,0.0) {Alice};
		\node[block] (b) at (-3.7,-1.5) {Decoder};
		\node[block] (o1) at (-6,-1.5) {Bob};
		\draw (c.east)  -- node[above] {$y_i$} (1.5,0);
		\path[line] (u1) -- node[near start, above] {$m$} (e1);
		\path[line] (e1.east) -- node[above] {$x_i$} (c.west);
		\path[line] (1.5,0) -- (1.5,-1.5) -- node[above] {$(y_1,\ldots, y_n)$} (b.east);
		\path[line] (1.5,0) -- (1.5,1.5) -- (d.east);
		\path[line] (d.west) -- node[above] {$(y_1,\ldots, y_{n_1})$ {\footnotesize if $i>n_1$} } (-4.2,1.5) -- ([xshift=-0.5cm]e1.north);
		\path[draw, dashed, rounded corners] (-4.8,2.3) -- (-4.8,-2.3) -- (1.9,-2.3) -- (1.9,2.3) -- cycle;
		\path[line] (b.west) -- node[near end, above] {$\hat{m}$} (o1);
		\node[draw, circle, minimum size=1mm, inner sep=0pt, outer sep=0pt, fill=darkgray] at (1.5,0) {};
		\end{tikzpicture}
		\caption{Two-stage coding over the Z-channel.}
		\label{fig::communication channel}
	\end{figure}
 
%
\subsection{Our contribution and methodology}
Our contribution can be split into three components. 

\textbf{1)} We show that for the Z-channel without feedback, an error-correcting code correcting a  fraction  $1/4+\epsilon$ of errors has size at most $O(\epsilon^{-3/2})$. This bound is shown by proper partitioning an arbitrary code into $O(\epsilon^{-1/2})$ almost constant-weight subcodes and showing that each subcode contains $O(\epsilon^{-1})$ codewords. 


\textbf{2)} We provide upper and lower bounds for the maximal cardinality of list-decodable codes for the Z-channel. Our lower bound is inspired by the probabilistic method, whereas the upper bound is based on the  double counting technique.

\textbf{3)} We describe a two-stage encoding strategy by combining two ideas. For the first batch of channel uses, we make use of a random constant-weight code for which we derive the list-decoding radius for all list sizes. The nature of the Z-channel and the constant-weight property of the code enable us to find the number of errors inflicted by the channel at the first stage. After this stage, we list decode the received string  to get a list of candidate messages.  Depending on the remaining noise, at the second stage we use either a code of small size which tolerates a large fraction of asymmetric errors, or a $\frac{1}{2}$-constant-weight code with a positive rate and large minimum distance. This strategy leads to the existential result. Our converse result also relies on the concepts of list-decodable codes and high-error low-rate codes. Thereby, we show that $R_{Z}^{(2)}(\tau)>0$ for $\tau<\tau_{\max}\coloneqq \max\limits_{0<w<1}\frac{w + w^3}{1+4w^3}\approx 0.44$ and $R_{Z}^{(2)}(\tau)=0$ for $\tau>\tau_{\max}$.
\begin{remark}
More general converse results on zero-rate list-decodable codes for the Z-channel  and order-optimal code constructions are presented in the parallel work~\cite{polyanskii2021codes}, where the bound $O(\epsilon^{-3/2})$ was first derived.
\end{remark}
\begin{remark}
The lower bound for list-decodable codes is derived by similar methods as in~\cite{blinovsky1986bounds}. The authors thank Yihan Zhang for showing the converse bound for list-decodable codes.
\end{remark}
\begin{remark}
We note that the methodology that is used for analyzing two-stage encoding schemes is
close to the ideas presented in works~\cite{chen2015characterization,chen2019capacity}, where the authors characterized the capacity of binary and non-binary online (or causal) channels.
\end{remark}
\subsection{Outline}
The remainder of the paper is organized as follows. In Section~\ref{sec::preliminaries}, we introduce the required notation and definitions. Section~\ref{sec::high error low rate codes} discusses high-error low-rate codes for the Z-channel. In Section~\ref{sec::list decodable codes}, we introduce the concept of list-decodable codes for the Z-channel and investigate lower and upper bounds on the maximal cardinality of such codes. Section~\ref{sec::2 stage algorithm} describes the suggested two-stage encoding algorithm and derives a Plotkin-type point for this problem. Finally, Section~\ref{sec::conclusion} concludes the paper.
\section{Preliminaries}\label{sec::preliminaries}
 We start by introducing some notation that is used throughout the paper. The set of integers from $m$ to $n$, $m\le i \le n$, is abbreviated by $[m,n]$ or simply $[n]$ if $m=1$. A vector of length $n$ is denoted by bold lowercase letters, such as $\x$,
	and the $i$th entry of the vector $\x$ is referred to as $x_i$. 	Given a binary vector $\x$, we define its support $\s(\x)$ as the set of coordinates in which the vector $\x$ has nonzero entries. By $\0$ and $\1$ denote the all-zero and the all-one vectors, respectively.
For $\x,\y\in\{0,1\}^n$, let $\Delta(\x,\y)$ denote the number of positions $i$ such that $x_i=1$ and $y_i=0$. We define the \textit{asymmetric distance}, written as $d_{Z}(\x,\y)$, to be  $2\max(\Delta(\x,\y),\Delta(\y,\x))$. The \textit{symmetric (Hamming) distance} $d_H(\x,\y)$ is then $\Delta(\x,\y)+\Delta(\y,\x)$. By $\wt(\x):=d_H(\x,\0)$ we abbreviate the \textit{Hamming weight} of $\x$. For $\x\in\{0,1\}^n$, the quantity $\wt(\x)/n$ stands for the \textit{normalized Hamming weight} of $\x$. Because of the relation $d_{Z}(\x,\y)=d_H(\x,\y)+|\wt(\x)-\wt(\y)|$, we have that for $\x,\y\in\{0,1\}^n$ with $\wt(\x)=\wt(\y)$, the asymmetric and symmetric distances between $\x$ and $\y$ coincide, i.e, $d_H(\x,\y)=d_{Z}(\x,\y)$.

An arbitrary subset $\C\subset\{0,1\}^n$ is called a \textit{code}.  The number of codewords in $\C$ is denoted by $|\C|$ and is called the \textit{size} of the code $\C$. The quantity $\log |\C|/n$ is called the \textit{rate} of $\C$. The code $\C$ is said to be \textit{$\omega$-constant-weight} if the  Hamming weight of each codeword is $\lfloor \omega n\rfloor$. The minimum asymmetric distance of the code $\C$, written as $d_{Z}(\C)$, is then defined as the minimum of $d_{Z}(\x,\y)$ over all distinct $\x$ and $\y$ from $\C$. Similarly, we define the minimum symmetric (Hamming) distance $d_{H}(\C)$ of the code $\C$.
It is known~\cite{kim1959single,varshamov1965theory} that a code $\C$ with $d_{Z}(\C)\ge d$ can correct up $t=\lfloor\frac{d-1}{2}\rfloor$ errors in the Z-channel. In what follows, we assume that $t:=\lfloor\frac{d-1}{2}\rfloor$.

By $A_H(n,t)$ ($A_{Z}(n,t)$) denote the maximum number of codewords in a code of length $n$ capable of correcting $t$ symmetric (asymmetric) errors.  We define $A_H(n,w,t)$ ($A_{Z}(n,w,t)$) as the maximum number of codewords in a constant-weight code with the Hamming weight $w$ and length $n$ capable of correcting $t$ symmetric (asymmetric) errors. It is readily seen that $A_H(n,w,t)=A_{Z}(n,w,t)$ as the definitions of symmetric and asymmetric distances coincide for constant-weight codes. 

Now we are in a good position to briefly recall some well-known result which appear to be useful for us. Define $w_0=w_0(n,t)$ to be $(n-\sqrt{n^2-4t n})/2$. It is known~\cite[Lemma]{bassalygo1965new} that if a code can correct $t$ symmetric errors, then it is a list-decodable code with a radius at most $w_0$ and the list size is polynomial in $n$. 
\begin{proposition}[Converse bound for codes with small constant weight]\label{th::list size when small radius} For $
t+1\le w \le w_0$, we have 
$$
A_H(n,w,t)=A_{Z}(n,w,t)\le 
\left\lfloor\frac{tn}{w^2 - (w-t)n}\right\rfloor.
$$
\end{proposition}
If $w>w_0(n,t)$, then $A_H(n,w,t)$ could be exponential. However, the exponential growth can be bounded by Theorem 2 from~\cite{levenshtein1971upper}.
\begin{proposition}[Converse bound for codes with large constant weight]\label{th::list size when large radius} For $
w_0< w \le n/2$, $n\to\infty$, we have 
\begin{align*}
\frac{\log A_H(n,w,t)}{n}&\le 
h\left(\frac{w}{n}\right) - h\left(\frac{w_0}{n}\right) + o(1),
\end{align*}
where $h(x)$ denotes the binary entropy function.
\end{proposition}
We state the classic Plotkin bound which was proved in~\cite{plotkin1960binary}. 
\begin{proposition}[Plotkin bound]\label{prop: classic plotkin bound} For $t>n/4$, it holds
  $$
  |A_H(n,t)|\le 2\left\lfloor \frac{2t+2}{4t+3 -n}\right\rfloor.
  $$
This implies that a code correcting a fraction $1/4 + \epsilon$ of symmetric errors contains at most $\lfloor 1 + 1/(4\epsilon)\rfloor$ codewords.
\end{proposition}

Finally, we mention an existential result. Constructions based on Hadamard matrices are shown~\cite{levenshtein1961application} to achieve the so-called Plotkin bound for symmetric errors. We will make use of a weak version of this result which follows from ~{\cite[Theorem 1]{levenshtein1961application}} and works for all parameters.
\begin{proposition}[Construction of high-error codes] \label{prop::construction for symmetric errors}
For any $\epsilon>0$ and $M\ge 1$, there is $n_0=n_0(M,\epsilon)$ such that for all $n>n_0$, there exists a code of size $M$ and length $n$ capable of correcting a fraction $\frac{M}{4M-2} - \epsilon$ of symmetric errors.

\end{proposition}
\section{High-error low-rate codes for the Z-channel}\label{sec::high error low rate codes}
In this section we discuss error-correcting codes for the case when the fraction of asymmetric errors is large. As the main result of this line of research, we prove that the cardinality of a code that corrects a $1/4+\epsilon$ fraction of asymmetric errors is at most $O(\epsilon^{-3/2})$.

	Let us introduce the notion of the maximum fraction of  correctable asymmetric errors for codes of a given size.
	\begin{definition}[Maximum fraction of correctable errors]\label{def::max correctable fraction}
		Given a positive integer $M$, define the quantity  $\tau_Z(M)$ to be the supremum of $\tau$ such that there exists a code of size $M$ that corrects a fraction $\tau$ of asymmetric errors.
	\end{definition}

For $M\ge 2$, define a binary matrix $D=D(M)$ with $\binom{M}{2}$ rows indexed by pairs from the set $[M]$ and $2^{M}$ columns. For $1\le i<j\le M$ and $k\in[2^M]$, the entry $D_{(i,j),k}$ equals $1$ iff the $i$th and $j$th entries in the binary representation of the integer $k$ are $0$ and $1$, respectively. 
By applying some linear programming arguments, the following statement on $\tau_Z(M)$ was proved in~\cite{borden1983low}.
\begin{lemma}[Follows from~\cite{borden1983low}]
For $M\ge 2$, the maximum fraction of correctable asymmetric errors $\tau_Z(M)$ satisfies
\begin{equation}\label{eq: maximization tau}
\tau_Z(M)^{-1} = \max_{(*)} \sum_{i=1}^{\binom{M}{2}}y_i,
\end{equation}
where the maximization $(*)$ is taken over all possible real vectors $\y$ of length $\binom{M}{2}$  such that each entry of $\y$ is non-negative and each entry of $\y D$ is at most $1$.
\end{lemma}
We compute $\tau_Z(M)$ for some small $M$ and depict these values in Table~\ref{table:1}. This table extends a similar one given in~\cite{borden1983low}.
\begin{remark}
Table~\ref{table:1} already disproves the claim from~\cite{borden1983low} saying that $\tau_Z(M)\ge 1/3$ for any $M$.  In particular, the mistake in~\cite{borden1983low} was made when the feasible region of the original linear program was relaxed and the author derived the wrong equation (4) in~\cite{borden1983low} using the correct one (3). 
\end{remark}
\begin{table}[t]
\centering
\caption{Maximum fraction of correctable asymmetric errors for codes of a given size}
 \begin{tabular}{||c | c ||c | c || c | c ||} 
 \hline
 $M$ & $\tau_Z(M)$ & $M$ & $\tau_Z(M)$ & $M$ & $\tau_Z(M)$\\ [0.5ex] 
 \hline\hline
 18 & $\frac{1083}{3467}$ & 13 & $\frac{18}{55}$  & 8 & $\frac{4}{11}$\\ 
  [0.3ex] 
 \hline
 17 & $\frac{712}{2263}$ & 12 & $\frac{1}{3}$ & 7 &  $\frac{3}{8}$ \\
  [0.3ex] 
 \hline
 16 & $\frac{1029}{3238}$ & 11 & $\frac{31}{92}$ &  5,6 & $\frac{2}{5}$ \\ 
  [0.3ex] 
 \hline
15 & $\frac{377}{1177}$ & 10 & $\frac{9}{26}$  & 3,4& $\frac{1}{2}$\\
  [0.3ex] 
 \hline

14 & $\frac{35}{108}$  & 9 & $\frac{13}{37}$  & 2&  $1$ \\ 
  [0.3ex] 
 \hline
\end{tabular}

\label{table:1}
\end{table}

Now we proceed with a trivial statement saying that there exist infinitely long codes of size $M$  correcting a fraction $\tau_Z(M)$ 
of asymmetric errors.
\begin{proposition}[Long high-error codes]\label{prop:trivial construction}
For any real $\epsilon>0$ and integer $M\ge 1$, there exists $n_0=n_0(\epsilon,M)$ such that for all $n>n_0$ there exists a code of size $M$ and length $n$ correcting a fraction $\tau_Z(M)-\epsilon$ of asymmetric errors.
\end{proposition}
In the following statement, we derive an upper bound on the maximal cardinality of a code capable of correcting a large fraction of asymmetric errors. The idea of the proof is to partition a code into $O(\epsilon^{-1/2})$ subcodes  according to the Hamming weight of codewords and prove that each subcode contains only $O(\epsilon^{-1})$ codewords.
\begin{lemma}[Plotkin-type bound for asymmetric error-correcting codes]\label{lem::plotkin bound z channel}
For $\epsilon>0$, any code correcting a fraction $1/4+\epsilon$  of asymmetric errors  contains $O(\epsilon^{-3/2})$ codewords. Furthermore, $\tau_{Z}(M)\to 1/4$ as $M\to\infty$. 
\end{lemma}
\begin{proof}
	Let $\C\subset \{0,1\}^n$ be a code correcting $(1/4+\epsilon)n$  asymmetric errors.  Let $\rho_i:=\frac{i}{2i+1}$. For simplicity of notation, we assume that $\rho_i n$ and $\epsilon n$ are integers. Define a subcode $ \C_i'\subset \C$ containing only codewords with weight in the range $[\rho_i n + 1,\rho_{i+1}n]$.  We append $(\rho_{i+1}-\rho_i)n-1$ extra coordinates to each codeword in $\C_i'$ such that the augmented codewords have the same weight $\rho_{i+1}n$. Note that this can be done in  different ways. From this point we assume that the code $\C_i'\subset\{0,1\}^{(1+\rho_{i+1}-\rho_i)n-1}$ contains only codewords with weight $\rho_{i+1}n$ and corrects $(1/4+\epsilon)n$ asymmetric errors. By Proposition~\ref{th::list size when small radius}, we have that
	\begin{align}
	&|\C_i'| \nonumber\\
	\le& \left\lfloor\frac{(1/4+\epsilon)n ((1+\rho_{i+1}-\rho_i)n-1)}{\rho_{i+1}^2n^2 - (\rho_{i+1} n - n/4-\epsilon n)((1+\rho_{i+1}-\rho_i)n-1)}\right\rfloor \nonumber\\
	\le&
	  \left\lfloor\frac{(1/4+\epsilon)(1+\rho_{i+1}-\rho_i)}{\rho_{i+1}^2 - (\rho_{i+1} - 1/4-\epsilon)(1+\rho_{i+1}-\rho_i)}\right\rfloor.\label{eq: first partitioning}
	\end{align}
	Note that $\rho_{i+1}^2 - (\rho_{i+1}-1/4)(1+\rho_{i+1}-\rho_i)= 0$ as  $\rho_i=i/(2i+1)$. Thus, $|\C_i'|\le \lfloor 1 + 1/(4\epsilon)\rfloor$.  
Let us take $i_0:= \lceil 1/\sqrt{3\epsilon} \rceil$. We form a subcode $\C_j''$ consisting of all codewords of $\C$ with weight in the range  $[\rho_{i_0}n+j\epsilon n + 1,\rho_{i_0}n+(j+1)\epsilon n]$, $j=0,1,2 \dots$. Again, we append $\epsilon n$ extra positions to each codeword in $\C_j''$ such that the augmented codewords have the same weight $\rho_{i_0}n+(j+1)\epsilon n$. From this point we assume that the code $\C_j''\subset\{0,1\}^{(1+\epsilon)n}$ contains only codewords with weight $\rho_{i_0}n+(j+1)\epsilon n$ and corrects $(1/4+\epsilon)n$ asymmetric errors.
Note that a fraction of correctable symmetric errors for these codewords is then $\frac{1/4+\epsilon}{1 + \epsilon}=\frac{1}{4}+\frac{3\epsilon}{4(1+\epsilon)}$. Then  by Proposition~\ref{prop: classic plotkin bound} we have
	$$
	|\C_j''|\le 
		\left\lfloor1+\frac{1+\epsilon}{3\epsilon} \right\rfloor 
		=
		 \left\lfloor\frac{(1+4\epsilon)}{3\epsilon}\right\rfloor.
	$$
Set $j_0:=\frac{3}{2\sqrt{3\epsilon}}$. Clearly,  each codeword of $\C$ has weight within one of the following intervals $[0,\rho_{i_0}n]$, $[\rho_{i_0}n+1, \rho_{i_0}n+j_0\epsilon n]$, $[\rho_{i_0}n+j_0\epsilon n+1,\ldots, n]$. Note that $n-\rho_{i_0}n-j_0\epsilon n < \rho_{i_0}n $ because of the choice  $j_0$. Thus, one can partition the set of codewords  having weight in the range $[\rho_{i_0}n+j_0\epsilon n+1,\ldots, n]$ into at most $i_0$ subcodes such that each of them contains only codewords with weight in the range $[(1-\rho_{i+1})n+1,(1-\rho_i)n]$ for some $i\in[0,i_0-1]$. Then it is easy to argue that the size of each subcode is at most $\lfloor 1+1/(4\epsilon)\rfloor$ (similar to the inequality~\eqref{eq: first partitioning}).
Therefore, we obtain
\begin{align*}
&|\C|\\
\le &2 i_0 \left\lfloor 1+\frac{1}{4\epsilon}\right\rfloor + j_0\left\lfloor\frac{(1+4\epsilon)}{3\epsilon}\right\rfloor
\\
\le &\left(1+\frac{1}{4\epsilon}\right) \left(\frac{2}{\sqrt{3\epsilon}}+2\right)+\frac{4\epsilon+1}{3\epsilon} \frac{3}{2\sqrt{3\epsilon}}
\\
\le
&\frac{1}{\epsilon\sqrt{3\epsilon}}+\frac{1}{2\epsilon}+\frac{4}{\sqrt{3\epsilon}}+2. 
\end{align*}
The above arguments imply that $\tau_Z(M) = 1/4+ O(M^{-2/3})$.  Proposition~\ref{prop::construction for symmetric errors} says that $\tau_Z(M)=1/4+\Omega(M^{-1})$. Hence, $\tau_Z(M)\to 1/4$ as $M\to \infty$.
\end{proof}
\section{List-decodable codes for the Z-channel}\label{sec::list decodable codes}
In this section, we discuss the concept of list-decodable codes for asymmetric errors. For other error models, e.g. symmetric errors, bounds on the maximal achievable cardinality of list-decodable codes have been
extensively studied in recent years~\cite{blinovsky1986bounds,polyanskiy2016upper,zhang2020generalized,alon2018list}.
We derive upper and lower bounds on the maximal cardinality of list-decodable codes for the Z-channel, which appear to be useful for providing a two-stage encoding algorithm.

For a point $\x\in\{0,1\}^n$ and an integer $t\in[n]$, define the \textit{$Z$-ball} with center $\x$ and radius $t$ as 
$$
B_{Z}(\x,t) := \{\y\in\{0,1\}^n:\ \Delta(\y,\x)\le t,\ \Delta(\x,\y)=0\}.
$$
\begin{definition}
	We say that a code $\C\subset \{0,1\}^n$ is $(t,L)_{Z}$-\textit{list-decodable} if for any $\x\in\{0,1\}^n$, the ball $B_{Z}(\x,t)$ contains at most $L$ codewords from $\C$.
\end{definition}
For any code $\C\subset \{0,1\}^n$ and any list size $L$, define $t_L(\C)$ to be the maximum integer $t$ such that $\C$ is $(t,L)_{Z}$-list-decodable. Define the \textit{normalized $L$-radius} of the code $\C$ as $\tau_L(\C):=t_L(\C) / n$. By $A_{Z}(n,t,L)$ denote the maximal cardinality of an $(t,L)_{Z}$-list-decodable code of length $n$. Define the \textit{maximal asymptotic rate} of $(t,L)_{Z}$-list-decodable codes to be
$$
R_{Z}(\tau,L):=\limsup_{n\to\infty}\frac{\log A_{Z}(n,\lfloor \tau n\rfloor ,L)}{n}.
$$
For $L+1$ points $\x^{(1)},\ldots,\x^{(L+1)}\in \{0,1\}^n$, 
define their \textit{average radius} by
\begin{align*}
	&rad(\x^{(1)},\ldots,\x^{(L+1)})
	\\:= &\frac{1}{L+1}\min_{\substack{\y\in \{0,1\}^n \\ \Delta(\y,\x^{(j)})=0\ \forall j\in[L+1]}} \sum_{i=1}^{L+1} \Delta(\x^{(i)},\y).
\end{align*}
Clearly, the minimum is achieved on the vector $\y$ whose support is the intersection of the supports of vectors $\x^{(i)}$. Note that if $rad(\x^{(1)},\ldots,\x^{(L+1)})> t$, then for any $\y\in\{0,1\}^n$, the ball $B_{Z}(\y,t)$ does not contain all $\x^{(i)}$ with $i\in[L+1]$.

Let $\C\subset\{0,1\}^n$ be a fixed code with cardinality $M$. We enumerate all codewords of this code such that $\C=\{\x^{(1)},\ldots,\x^{(M)}\}$. For a subset $\L\in\binom{[M]}{L+1}$, define $\y_\L$ 
to be a binary vector whose support is the intersection of the supports of  $\x^{(i)}$ with $i\in \L$. 
For $\L\in\binom{[M]}{L+1}$ with $\L=\{i_1,\ldots, i_{L+1}\}$, let $rad(\L,\C):=rad(\x^{(i_1)},\ldots,\x^{(i_{L+1})})$. For a code $\C\subset\{0,1\}^n$, it is natural to define the \textit{ average $L$-radius} of the code $\C$, written as $rad_{L}(\C)$, to be the minimum $rad(\L,\C)$ over all $\L\in \binom{[M]}{L+1}$. Observe that if $rad_L(\C)>t$, then $\C$ is $(t,L)_{Z}$-list-decodable and $\tau_L(\C)\ge t/n$.
\subsection{Lower Bound on $R_{Z}(\tau,L)$}
\begin{theorem}[Random coding bound for list-decodable codes] \label{th::random coding bound}
Let $w$ be a fixed real number such that $0<w<1$ and $L_{up}$ be a fixed positive integer.	For any $R$, $0<R<h(w)$, and any sequence of positive integers $\{n_i\}$ with $\lim\limits_{i\to \infty} n_i = \infty$, there exist constant-weight codes $\{\C_i\}$, $\C_{i}\subset \{0,1\}^{n_i}$, such that for $i\to \infty$, it holds that 
\\		1) the rate $R(\C_i)=\frac{\log |\C_i|}{n_i}\ge R(1+o(1))$,
\\   2) the normalized $L$-radius 
$$\tau_L(\C_i)\ge \frac{rad_L(\C_i)}{n_i}\ge \tau^*(R,L,w) (1+o(1))
$$
for any $L\in[L_{up}]$, where
		$$
		\tau^*(R,L,w) := \sup\limits_{\substack{h_L>0 \text{ subject to}\\ g(h_L,L,w)-h_L \delta(h_L,L,w) \ge RL }} \delta(h_L,L,w),
		$$
		\begin{align*}
		g(h_L,L,w) := &-\log\left(\left(w2^{\frac{-h_L}{L+1}}+1-w\right)^{L+1}\right.
		\\
		&+w^{L+1}\left(1-2^{-h_L}\right)\bigg),
		\end{align*}
		\begin{align*}
		&\delta(h_L,L,w):=2^{g(h_L,L,w)}
		\left(w2^{\frac{-h_L}{L+1}} \right.\\
		&\quad\times \left.\left(w2^{\frac{-h_L}{L+1}}+1-w\right)^{L}-w^{L+1}2^{-h_L}\right),
		\end{align*}
\\		3) the normalized Hamming weight of all codewords of $\C_i$ is $w(1+o(1))$.
\end{theorem}
	\begin{remark}
	Given $R$ and $L$, define $\underline{\tau}(R,L):=\sup\limits_{0\le w\le 1} \tau^*(R,L,w)$. By taking the inverse function to $\underline{\tau}(R,L)$ we derive lower bounds on $R_{Z}(\tau,L)$ for $L=1,2,3,10$ and plot them in Figure~\ref{fig::asymptotic rate low bounds list}. We note that the lower bound on $R_{Z}(\tau,1)$ coincides with the well-known Gilbert-Varshamov bound.
\end{remark}
\input{plotListLow}
\begin{proof} Consider a random code $\C$ of size $M=2^{Rn}$ and length $n$ whose codewords are taken independently from $\{0,1\}^n$  such that each bit is $1$ with probability $w$ and $0$ with probability $1-w$. It may happen that $\C$ contains several copies of the same word.  Let $t_L$ be a positive integer which will be specified later. For any $\L\in\binom{[M]}{L+1}$, define a random variable $Y_\L$ that takes $0$ if $rad(\L,\C)> t_L$ and $1$ otherwise. We can think about $Y_\L$ as a function indicating that the set $\L$ is \textit{bad}. Indeed, if $Y_\L=1$, then it might happen that there is a Z-ball with radius $t_L$ containing the codewords indexed by $\L$. On other hand, if $Y_\L=0$, then the codewords indexed by $\L$ cannot lie in a Z-ball with radius $t_L$.  Let the total number of bad sets be defined as 
$$
	Y_L:=\sum_{\L\in \binom{[M]}{L+1}} Y_\L.
	$$ 
	If we throw away from $\C$ one codeword from each bad $\L \in\binom{[M]}{L+1}$, we obtain a $(t_L,L)_{Z}$-list-decodable code. Note that random variables $Y_{\L}$ for all $\L\in \binom{[M]}{L+1}$ have the same distribution. From the Markov inequality $\Pr\{Y_L\ge 2^{L} \E [{Y}_L]\}\le \frac{1}{2^L}$, it follows that with a positive probability there exists a code $\C'$ with cardinality 
	\begin{align*}
	&M- \sum_{L=1}^{L_{up}} 2^{L} \E [{Y}_L]\\
	= &M-\sum_{L=1}^{L_{up}} 2^{L} \binom{M}{L+1}\E[{Y}_{\{1,\ldots,L+1\}}] \\
	\ge &M - \sum_{L=1}^{L_{up}} M^{L+1} \E[{Y}_{\{1,\ldots,L+1\}}] 
	\end{align*} 
	which is  $(t_L,L)_{Z}$-list-decodable for all $L\in[L_{up}]$. In what follows, we find the conditions sufficient for
	\begin{equation}\label{eq::condition on size and t}
	M^{L+1} \E[{Y}_{\L}] \le M/2^{L+1}  \iff  2^{RnL} \E[{Y}_{\L}]\le 1 / 2^{L+1},
	\end{equation}
	which will imply the existence of a code $\C'$ of size $M/2$. Given $L$, let us fix $\L=\{1,\ldots,L+1\}$.
	It remains to estimate $\E [{Y}_{\L}]$ and choose $t_L$ appropriately. To this end, we consider a random variable $rad(\L,\C)$ that can be represented as a sum of $n$ independent copies of a random variable $\xi_{L}$ with
	$$
	\Pr\left\{\xi_{L}=\frac{i}{L+1}\right\}=
	\begin{cases}
	w^{L+1}+(1-w)^{L+1},\quad&\text{for }i=0, \\
	\binom{L+1}{i}w^{i}(1-w)^{L+1-i},\quad &\text{for } i\in[L].
	\end{cases}
	$$
	Note that $\E[ {Y}_{\L}] = \Pr\{rad(\L,\C) \le t_L\}$. By the Chernoff bound, for any $h_L,\alpha_L>0$, the random variable $\eta_L:=rad(\L,\C)$ deviates from its expected value $\E[{\eta_L}]$ with probability
	\begin{align*}
	&\Pr\{\E[{\eta_L}] - \eta_L  \ge \alpha_L n\}\\ 
	\le &2^{-\alpha_L h_L n} \E[ 2^{h_L(\E[\eta_L] - \eta_L)}] \\
	= &2^{h_L\E [\eta_L] -\alpha_L h_L n } (\E[2^{-h_L\xi_{L}}])^n.
	\end{align*}
	First, observe that 
	\begin{align*}
	\E[\eta_L] &= n \E[\xi_{L}] = \frac{n}{L+1} \sum_{i=1}^{L}i\ \binom{L+1}{i}w^i(1-w)^{L+1-i}\\
	&= nw (1-w^L).
	\end{align*}
	Second, we check that 
	\begin{align*}
	&\E[2^{-h_L\xi_{L}}]=w^{L+1}+(1-w)^{L+1}\\
	&\quad+\sum_{i=1}^{L} 2^{-\frac{h_L i}{L+1}}\binom{L+1}{i}w^i (1-w)^{L+1-i}\\
	=&(w2^{\frac{-h_L}{L+1}}+1-w)^{L+1} + w^{L+1}(1-2^{-h_L}).
	\end{align*}
	Thus, we obtain that
	\begin{align*}
	&\Pr\{\eta_L \le \E[{\eta_L}] - \alpha_L n\} \\
	\le &2^{-n\left(h_L(\alpha_L - (w-w^{L+1}))-\log\left(\E [2^{-h_L\xi_{L}}] \right)\right)}.
	\end{align*}
	Define $g(h_L,L,w):=-\log\left(\E [2^{-h_L\xi_{L}}] \right)$. To get a stronger estimate, we optimise the right-hand side of the above inequality over the choice of $h_L>0$. It is clear that the minimum is attained at $h_L$ satisfying
	\begin{align*}
&(w-w^{L+1}) - \alpha_L 
= \frac{\partial g(h_L,L,w) }{\partial h_L}
\\
=&2^{g(h_L,L,w)}
		\left(w2^{\frac{-h_L}{L+1}}\left(w2^{\frac{-h_L}{L+1}}+1-w\right)^{L}-w^{L+1}2^{-h_L}\right).
	\end{align*}
	Denote the right-hand side of the above equation by $\delta(h_L,L,w)$. In the following, we set  
	$$
	\alpha_L= \alpha_L(h_L):=(w-w^{L+1}) - \delta(h_L,L,w).
	$$
	Let 
	$$
	t_L:=\E[\eta_L] -n  \alpha_L = n \delta(h_L,L,w).
	$$
	Then we derive
	\begin{align*}
	\E[Y_{\L}] &=\Pr\{\eta_L\le t_L\}\\
	&=\Pr\{\eta_L\le \E[\eta_L]-\alpha_L n\}\\
	&\le 2^{-n(g(h_L,L,w)-h_L\delta(h_L,L,w))}.
	\end{align*}
	To have~\eqref{eq::condition on size and t}, we need to have $h_L$ such that
	$$
	g(h_L,L,w) - h_L \delta(h_L,L,w) > RL + (L+1)/n.
	$$
	Given $0<R<1$, we want to maximize the quantity $t_L$. This asymptotically implies that 
	$$
	rad_L(\C')\ge n \tau^*(R,L,w)(1+o(1)).
	$$
	
	 To guarantee that all words $\C'$ have the same normalized Hamming weight close to $w(1+o(1))$ as $n\to\infty$, we can additionally pre-process the set $\C$. This can be done by the standard techniques, e.g., by the Hoeffding inequality, a large fraction of generated words has an appropriate Hamming weight and, thus, we can throw all other words away to get a proper $\C'$.  This completes the proof.
\end{proof}
\begin{corollary}\label{cor::important lower bound}
	Let  $w$ be a real number such that $0< w < 1$ and $L_{\text{up}}$ be a positive integer. Fix any $\epsilon$ such that $0<\epsilon<w -w^2$. Then there exists $R=R(\epsilon, L_{up})>0$ and $n_0=n_0(\epsilon, L_{up})$ such that the following holds: for any $n\ge n_0$, there exists a $\overline{w}$-constant-weight code of length $n$ and size at least $2^{Rn}$ whose normalized $L$-radius is at least $w - w^{L+1}-\epsilon$ for all $L\in[L_{\text{up}}]$. Moreover, $\overline{w}\in(w - \epsilon, w+\epsilon)$. 
\end{corollary}
\begin{proof}
We apply Theorem~\ref{th::random coding bound} and use the notation introduced in that statement.	Define $f(h_L,L,w):=g(h_L,L,w)-h_L\delta(h_L,L,w)$. Note that $g(0,L,w)= 0$, $\delta(0,L,w)=w - w^{L+1}$ and, thus, $f(0,L,w)=0$. 

First, we shall prove that $f(h_L,L,w)>0$ for small enough $h_L>0$. Since the function $\delta(h_L,L,w) = \frac{\partial g}{\partial h_L}$, the derivative $\frac{\partial f}{\partial h_L}=-h_L\frac{\partial \delta}{\partial h_L}$ and it suffices to check that $\frac{\partial \delta}{\partial h_L} < 0$ at point $0$. Using
a symbolic computation package it is easy to verify that
	for any $w\in(0,1)$, the derivative $\frac{\partial \delta}{\partial h_L}$ at point $0$ has the same sign as
	\begin{align*}
	&\frac{-w-Lw^2}{L+1}+w^{L+1}+(w-w^{L+1})^2\\
	=&
	\frac{-w+w^2}{L+1} + w^{L+1}(1-2w + w^{L+1})\\
	<&\frac{-w(1-w)}{L+1} + w^{L+1}(1-w)^2
	\\
	=&(1-w)\left(w^{L+1}-w^{L+2}-\frac{w}{L+1}\right)
	\\
	=&(1-w)w^{L+1}\left(1-w-\frac{w^{-L}}{L+1}\right)
	\\
	<&\,
0.
	\end{align*}
	Since the function $f$ is continuously differentiable, the above arguments yield that $f(h_L,L,w)$ is monotonously increasing when $h_L\in [0, h_L']$ for some real number $h_L'>0$. 
	Let $D_{L,w}(R)$ denote the region of feasible $h_L$, i.e., $D_{L,w}(R):=\{h_L>0: f(h_L,L,w)\ge R L\}$. Then the infimum of the set $D_{L,w}(R)$ converges to $0$ as $R\to 0$. By definition,
	$$
		\tau^*(R,L,w) = \sup\limits_{h_L\in D_{L,w}(R)} \delta(h_L,L,w).
	$$
	Therefore,  
	$$
	\lim\limits_{R\to 0}\tau^*(R,L,w)\ge \delta(0,L,w) = w - w^{L+1}.
	$$
	Since all the functions used in the definition of $\tau^*(R,L,w)$ are continuous, the required statement follows.
\end{proof}
\begin{remark}
 For a fixed positive integer $L$, Corollary~\ref{cor::important lower bound} implies the existence of positive-rate $(\tau n, L)_Z$-list-decodable codes for any $\tau< \max\limits_{0<w<1}(w - w^{L+1}) =\frac{L}{(L+1)^{\frac{L+1}{L}}}$. For $L=1,2,3,10$, one can see in Figure~\ref{fig::asymptotic rate low bounds list} that $\frac{L}{(L+1)^{\frac{L+1}{L}}} = 0.25, 0.385, 0.473, 0.715$. For $L\to\infty$, the largest relative list-decoding radius of exponential-sized codes for the $Z$-channel converges to $1$. Recall that for symmetric errors, a similar limit is $1/2$.
\end{remark}
\subsection{Upper Bounds on $R_{Z}(\tau,L)$}
We prove a Plotkin-type bound on the number of codewords in a list-decodable code for the Z-channel.
\begin{lemma}[Plotkin-type bound for list-decodable codes]\label{lem::plotkin bound list decoding}
	Let $\C\subseteq \{0,1\}^n$ be a code of size $M$ whose codewords have the  Hamming weight $w n$. If $\C$ is $(\tau n,L)_{Z}$-list-decodable with $\tau>w - w^{L+1}$, then 
	$$
	\frac{M^L}{(M-1)\ldots (M-L)}\ge \frac{\tau}{w - w^{L+1}}.
	$$
This implies that for any $\epsilon>0$ and $n\to\infty$, the rate of $w$-constant-weight codes of length $n$ with the relative $L$-radius $w-w^{L+1}+\epsilon$ vanishes.
\end{lemma}
\begin{proof}
     Let $M$ denote the number of codewords in the code $\C=\{\x^{(1)},\ldots,\x^{(M)}\}$. For a multiset $\mathcal{L}=\{i_1,\ldots,i_{L+1}\}\subset[M]$ of size $L+1$, define $\y_{\mathcal{L}}$ to be a word whose support is the intersection of supports of $\x^{(i_j)}$, $j\in[L+1]$. To prove an upper bound on $M$, we provide a standard double counting arguments. Consider the summation
\begin{equation}\label{eq::summation}
\sum_{\substack{\mathcal{L}\subset[M]^{L+1}
		\\ \mathcal{L}=(i_1,\ldots,i_{L+1})}}\sum_{j=1}^{L+1}d_H(\x^{(i_j)},\y_{\mathcal{L}}).
\end{equation}
 Note that if all elements of $\mathcal{L}$ are distinct, then the constant weight and the $(\tau n,L)_{Z}$-list-decodability property imply that $\Delta(\x^{(i_j)},\y_{\mathcal{L}})=d_H(x^{(i_j)},\y_{\mathcal{L}})\ge \tau n$ for all $j\in[L+1]$. Thus, the summation is at least $M(M-1)\dots (M-L) (L+1) \tau n$. On the other hand, we have
\begin{align*}
&\sum_{\substack{\mathcal{L}\subset [M]^{L+1}
		\\ \mathcal{L}=(i_1,\ldots,i_{L+1})}}\sum_{j=1}^{L+1} d_H(\x^{(i_j)},\y_{\mathcal{L}}) 
		\\= &(L+1)\sum_{\substack{\mathcal{L}\in [M]^{L+1}
		\\ \mathcal{L}=(i_1,\ldots,i_{L+1})}}\left(w n - \sum_{k=1}^n\prod_{j=1}^{L+1} \mathbb{1}\{x^{(i_j)}_k=1\}\right)
\\=&(L+1)M^{L+1} w n- (L+1)\sum_{k=1}^n S_k^{L+1},
\end{align*}
where $S_k$ denotes the number of codewords having $1$ at position $k$. Recall that $\sum_{k=1}^n S_k= M w n $ and $0\le S_k\le M $. 
It is easy to check that for the integer vector $\v=(S_1,\ldots,S_n)$, the norm inequality  $\left\|\v\right\|_p\le n^{1/p-1/q}\left\|\v\right\|_q$ with $0<p<q$ implies
\begin{align*}
\sum_{k=1}^n S_k^{L+1} &= \left(\left\|\v\right\|_{L+1}\right)^{L+1}
\\ &\ge (\left\|\v\right\|_{1}n^{1/(L+1) - 1})^{L+1} 
\\
&\ge (Mw n)^{L+1} n^{-L}.
\end{align*}
Finally, combining lower and upper bounds on the summation~\eqref{eq::summation} yield
$$
(M-1)\dots (M-L) \tau \le M^L w  - M^L w^{L+1}.
$$
This completes the proof.
\end{proof}
\section{Two-stage encoding algorithm}\label{sec::2 stage algorithm}
In Section~\ref{ss::encoding strategy}, we present a two-stage encoding algorithm that combines random list-decodable codes from Section~\ref{sec::list decodable codes} and high-error low-rate codes described in Section~\ref{sec::high error low rate codes}. In Section~\ref{sec: plotkin type point for two-stage encoding strategies}, we precisely characterize when exponential-sized (or positive-rate) codes exist for the two-stage model. 
\subsection{Encoding strategy}\label{ss::encoding strategy}
Fix a positive integer $L_{\text{up}}$, and real numbers $\epsilon>0$ which will be specified later. Let $\tau$ denote the fraction of errors, and $n$ be the total number of channel uses. For some $\alpha$, $0<\alpha <1$, define integers $n_1:=\alpha n$ and $n_2:=(1-\alpha)n$ which correspond to the number of channel uses at the first and second stages. Let $M$ be the total number of messages; $m\in[M]$ be a message that the sender wishes to send;  $\y_1$ denote the received string after the first stage; $\x_1=\x_1(m)$ and $\x_2=\x_2(m,\y_1)$ be strings transmitted by the sender at the first and second stages.  Let $w$ be a weight real parameter such that $0<w < 1$ and $R_1,R_2$ denote the rate of a code used at the first stage and the second stage, which will be specified later.  Define $M:=2^{R_1 n_1}=2^{R_1\alpha n}$. 

\textbf{First stage:}
By Theorem~\ref{th::random coding bound}, for any $\epsilon>0$, there exists a sufficiently large $n^*(w,\epsilon,L_{up}, R_1)$ such that for all $n_1>n^*(w,\epsilon,L_{up}, R_1)$, there exists a $w_1$-constant-weight code $\C\subset\{0,1\}^{n_1}$ of size $|\C|=M=2^{R_1n_1}$, which has the normalized $L$-radius at least $\tau^*(R_1,L,w)-\epsilon$ for any $L\in [L_{up}]$, where the weight parameter  $w_1\in(w-\epsilon, w+\epsilon)$.  Then for any message $m\in [M]$, the sender transmits the $m$th codeword, written as $\x_1$, of this code $\C$. 
Suppose that $\tau_1 n_1$ errors occur at the first $n_1$ channel uses and, hence, at most $\tau n-\tau_1 n_1=:\tau_2 n_2$ errors will happen in the remaining $n_2$ channel uses. Since an error may happen only when a one is transmitted, the received word, denoted as $\y_1$, has the Hamming weight $w_1 n_1 - \tau_1 n_1$.  Thus, the values $\tau_1 n_1$ and $\tau_2 n_2$ can be easily computed by the receiver. 

\textbf{Second stage:}  To describe the process of encoding at the second stage, we distinguish two cases.

\textit{1st case:} It holds that $0\le \tau_1 \le \tau^*(R_1,L_{up},w)-\epsilon$. Thus, $\tau^*(R_1,L-1,w)-\epsilon< \tau_1\le \tau^*(R_1,L,w)-\epsilon$ for some $L\in[L_{up}]$\footnote{For simplicity of notation, we make the assumption that for $L=0$, $0\le R\le 1$ and $0\le w\le 1$, $\tau^*(R,L,w)=0$}. In this case,  both the sender and the receiver can reconstruct up to $L$ candidate messages based on the output $\y_1$. To distinguish the original message from the $L$ candidates, the sender uses a high-error low-rate code of size $L$ from Proposition~\ref{prop:trivial construction} at the second stage. If $\tau_2 \le \tau_Z(L)-\epsilon$, the receiver decodes the message correctly. Observe that $\tau = \alpha\tau_1+ (1-\alpha)\tau_2$. Combining the above arguments, we come to the condition which is sufficient for error-free decoding in the first case
\begin{equation}\label{eq: first cond}
\tau\le \alpha(\tau^*(R,L-1,w) - \epsilon) + (1-\alpha) (\tau_Z(L)-\epsilon), \quad \forall L\in [L_{up}].
\end{equation}
 
\textit{2nd case:} It holds that $\tau^*(R_1,L_{up},w) -\epsilon < \tau_1 \le \min(w_1,\tau/\alpha)$.  In this case, the number of candidate messages consistent with the output $\y_1$ can be exponential in $n_1$. Now we estimate the corresponding exponent. To this end, we need to compute the number of codewords of $\C$ whose supports share the support of $\y_1$. Define $S:=\s(\y_1)$ with $|S|=w_1 n_1-\tau_1 n_1$. Form the shortened code $\C'$ that  contains all possible $\z\in\{0,1\}^{n_1(1-w_1+\tau_1)}$ with the property:  there exists $\x=\x(\z)\in\C$ such that $S\subset \s(\x)$  and $\z = \x|_{[n_1]\setminus S}$.
    Clearly, $\C'$ is $((\tau^*(R_1,L,w)-\epsilon)n_1,L)_Z$-list-decodable and contains codewords whose normalized Hamming weight is $\tau_1/(1-w_1+\tau_1)$. Taking into account the property that $\C'$ can correct up to $(\tau^*(R_1,1,w)-\epsilon)n_1$ errors, we conclude by Proposition~\ref{th::list size when large radius} that
    \begin{align*}
    &\frac{\log |\C'|}{n_1(1-w_1 +\tau_1)} \le h\left(\frac{\tau_1}{1-w_1+\tau_1}\right) \\ &- h\left(\frac{1-\sqrt{1-4(\tau^*(R_1,1,w)-\epsilon) / (1+w_1-\tau_1)}}{2}\right) +o(1).
    \end{align*}
    Define $R_2=R_2(\alpha,\tau_1,w,w_1, R_1,\epsilon)$ as follows
    \begin{align}\nonumber
    &R_2(\alpha,\tau_1,w,w_1, R_1,\epsilon)\\
    := &\frac{\alpha (1-w_1+\tau_1)}{1-\alpha} \left( h\left(\frac{\tau_1}{1-w_1+\tau_1}\right)  \right. \\ -&\left.h\left(\frac{1-\sqrt{1-4(\tau^*(R_1,1,w)-\epsilon) / (1+w_1-\tau_1)}}{2}\right) \right). \label{eq: rate at the second stage}
    \end{align}
    At the second stage, the sender transmits a codeword of a random code with rate $R_2$ whose codewords have the normalized Hamming weight $1/2$. The receiver decodes the message correctly if $\tau_2\le \tau^*(R_2,1,1/2)-\epsilon$. Since $\tau = \alpha\tau_1+ (1-\alpha)\tau_2$, we come to the following condition allowing error-free transmission
 \begin{equation}\label{eq: second cond}
		\tau \le \alpha (\tau^*(R_1,L_{up},w)- \epsilon) + (1-\alpha )(\tau^*(R_2,1,1/2)-\epsilon).
	\end{equation}

If both conditions~\eqref{eq: first cond} and~\eqref{eq: second cond} are satisfied for all possible  $\tau_1$ with $0\le \tau_1 \le \min(w_1,\tau/\alpha)$, then the proposed encoding scheme transmits $2^{R_1\alpha n}$ messages and can correct up to $\tau n$ asymmetric errors. By taking $\epsilon\to 0$, we derive the following statement.
\begin{theorem}\label{th: lower bound for two-stage encoding schemes}
For any positive integer $L_{up}$, the maximal asymptotic rate of two-stage error-correcting codes for the Z-channel satisfies
$$
R_{Z}^{(2)}(\tau)\ge \sup_{0\le w\le 1} \sup_{0\le \alpha\le 1}\sup_{\substack{0\le R \le h(w)\\ \text{subject to }(*)}} \alpha R,
$$
where the condition $(*)$ means that for any $x\in(0,\min(w,\tau/\alpha))$, it holds
\begin{enumerate}
    \item if $\tau^*(R,L-1,w)\le x\le \tau^*(R,L,w)$ for $L\in[L_{up}]$, then $(\tau-\alpha x)/(1-\alpha)\le \tau_Z(L)$.
    \item if $\tau^*(R,L_{up},w)\le x \le \min(w,\tau/\alpha)$, then $(\tau-\alpha x)/(1-\alpha) \le \tau^*(R_2,1,1/2)$ for $R_2$ being computed as $R_2(\alpha,x,w,w,R,0)$ in~\eqref{eq: rate at the second stage}.
\end{enumerate}

\end{theorem}
Unfortunately, we don't know how to get a closed form for this lower bound. But, using Table~\ref{table:1}, we are able to compute this bound numerically with $L_{up}=18$. In Figure~\ref{fig::asymptotic rate}, we compare our results for two-stage  encoding schemes to the known results for one-stage (non-adaptive) and fully adaptive error-correcting codes for the Z-channel. Recall that by~\cite{bassalygo1965new,borden1983low}, the asymptotic rate of one-stage codes correcting a fraction $\tau$ of asymmetric and symmetric errors is the same. Thereby, for the non-adaptive setting, one can use the Gilbert-Varshamov lower bound~\cite{varshamov1957estimate,gilbert1952comparison} and the McEliece-Rodemich-Rumsey-Welch upper bound~\cite{mceliece1977new} on the rate $R_{Z}(\tau,1)$. The work~\cite{deppe2020coding} established a lower bound on the rate for fully adaptive encoding strategies. One can check that exponential-sized (or positive-rate) codes exist (i) when $\tau<0.25$ for the non-adaptive setting, (ii) when $\tau<\tau_{\max}\approx 0.44$ for the two-stage setting (see Theorem~\ref{th: Plotkin point for two-stage schemes}), (iii) when $\tau<1$ for the fully adaptive regime. 
\input{plot.tex}
\subsection{Analysis of the Plotkin-type point}\label{sec: plotkin type point for two-stage encoding strategies}

Define 
	\begin{equation}\label{eq::tau max}
		\tau_{\max}:=\max\limits_{0<w<1}\frac{w + w^3}{1+4w^3}\approx 0.44.
	\end{equation}
	Let $w_{\max}\approx 0.66$ be the argument achieving the maximum in the above equation and $\alpha_{\max} := (1+4w_{\max}^3)^{-1}\approx 0.46$.

\begin{theorem}\label{th: Plotkin point for two-stage schemes}
 The rate $R_{Z}^{(2)}(\tau)>0$ for all $\tau<\tau_{\max}$ and  $R_{Z}^{(2)}(\tau)=0$ for $\tau>\tau_{\max}$.
\end{theorem}
\begin{proof}
To show that the rate $R_{Z}^{(2)}(\tau)$ is positive for all $\tau<\tau_{\max}$, we make use of the encoding strategy described in Section~\ref{ss::encoding strategy}. Note that this algorithm depends on the parameters $L_{up}$, $R_1$ and $\epsilon$, but we let $L_{up}\to\infty$ and $R_1=\epsilon\to 0$ in our further analysis. Recall that by Corollary~\ref{cor::important lower bound}  
$$
\lim\limits_{\epsilon\to 0}\tau^*(\epsilon,L,w)\ge w-w^{L+1}.
$$
Then it is easy to verify that the conditions~\eqref{eq: first cond} and~\eqref{eq: second cond} imply, that there exists a positive-rate two-stage code for any fraction of asymmetric errors less than $\tau_{\max}(w,\alpha)$, where  $\tau_{\max}(w,\alpha)$ is defined as the supremum over $\tau\ge 0$ such that
\begin{equation}\label{eq::system}
		\begin{cases}
			\tau \le  \alpha (w - w^{L}) + (1-\alpha)\tau_{Z}(L),\quad \forall L\ge 1,\\
			\tau \le \alpha w + (1-\alpha)/4.
		\end{cases}
\end{equation}
Define $\tau_{\max}(w):=\sup\{\tau_{\max}(w,\alpha): \ 0< \alpha< 1 \}$. 

	To conclude, we solve the optimization problem~\eqref{eq::system}. We first omit several inequality constraints and find a solution in the relaxed problem. Then we show that the omitted constraints are not violated for the obtained solution. Let $\hat \tau_{\max}(w)$ be the supremum of $\tau\ge 0$ taken over all possible $\alpha$ with $0\le \alpha \le 1$ subject to the first inequality with $L=3$ and the second inequality in~\eqref{eq::system}. Clearly, $\hat \tau_{\max}(w)\ge \tau_{\max}(w)$. Since $\tau_{Z}(3)=1/2$ (cf. Table~\ref{table:1}), we obtain 
	$$
	\hat \tau_{\max}(w)=(w+w^3)/(1+4w^3)
	$$ 
	for $w\ge 1/4$ and the supremum is attained at $\alpha=(1+4w ^3)^{-1}$. For $w\le 1/4$, $\hat \tau_{\max}(w) = 1/4$.   Observe that $\tau_{\max}=\max \{\hat \tau_{\max}(w):\ 0<w<1\}$ by definition~\eqref{eq::tau max} and this maximization is attained at $w=w_{\max}$. However, we are interested in $\sup \{\tau_{\max}(w):\ 0<w< 1\}$. To prove that this supremum equals $\tau_{\max}$, it suffices to prove that for $\alpha = \alpha_{\max}$ and $w=w_{\max}$, all inequalities with $L\neq 3$ in~\eqref{eq::system} are satisfied. In other words,  it remains to show that for all $L\neq 3$, it holds that
	\begin{align}
		\tau_{\max} \le \alpha_{\max}(w_{\max}-w_{\max}^{L})+(1-\alpha_{\max}) \tau_{Z}(L).\label{eq::remains}
	\end{align}
	Recall that $\alpha_{\max}=(1+4w_{\max}^3)^{-1}$ and $\tau_{\max}=\alpha_{\max}(w_{\max}+w_{\max}^3)$. After simple algebraic manipulations, we derive that the inequality~\eqref{eq::remains} is equivalent to
	$$
	\tau_{Z}(L) \ge 1/4 + w_{\max}^{L-3}/4.
	$$
	Using Table~\ref{table:1}, we check the validity of this inequality for small $L\le 10$. For larger $L$, we apply Proposition~\ref{prop::construction for symmetric errors} saying that $\tau_{Z}(L)\ge L/(4L-2)$. Thus, it suffices to check
	$$
	\frac{L}{4L-2} \ge 1/4 + w_{\max}^{L-3}/4. 
	$$
	Note that $w_{\max}\approx 0.661< 2/3$. By simplifying the above inequality, we obtain
	$$
	\frac{1}{2L-1} \ge (2/3)^{L-3}\ge w_{\max}^{L-3}.
	$$
	The latter holds for $L\ge 11$ and, thus, the inequality~\eqref{eq::remains} is correct. This implies that $\tau_{\max} = \tau_{\max}(w_{\max})$.
	Therefore, for any $\tau<\tau_{\max}$ and sufficiently large $n$, the proposed two-stage code with $w=w_{\max}$ and $\alpha=\alpha_{\max}$ has exponential size and corrects a fraction $\tau$ of asymmetric errors.

Now we turn to prove the converse result. Let $\{n^{(i)}\}$, $\{M^{(i)}\}$ and $\{\alpha^{(i)}\}$ be some infinite sequences of integers  and real numbers such that  $0<\alpha^{(i)}<1$, $n^{(i)}\to\infty$ and $\lim \frac{\log M^{(i)}}{n^{(i)}}>0$ as $i\to\infty$. Suppose that there exists  a series of two-stage encoding schemes such that: the $i$th two-stage code of size $M^{(i)}$ and length $n^{(i)}$ corrects a fraction $\tau^{(i)}$ of  asymmetric errors, and the first stage requires $\alpha^{(i)} n^{(i)}$ channel uses. Note that we can find at least $M^{(i)}/(n^{(i)}+1)$ messages that are encoded into strings of the same Hamming weight at the first stage. Clearly, $\lim\limits_{i\to\infty} \frac{\log M^{(i)}/(n^{(i)}+1)}{n^{(i)}}>0$. Thus, we may assume that the $i$th encoder transmits only strings with the normalized Hamming weight $w^{(i)}$ at the first stage for some real number $0\le w^{(i)}\le 1$.
	
	We shall prove that $\lim \sup \tau^{(i)}\le \tau_{\max}$. Toward a contradiction, assume that $\lim \sup \tau^{(i)} = \hat \tau > \tau_{\max}$. By the Bolzano-Weierstrass theorem, there exist real numbers $\hat\alpha$ and $\hat w$ such that for some infinite sequence of indices $\{i_j\}$, we have
	\begin{align*}
		&\lim_{j\to\infty} n^{(i_j)}=\infty,\quad \lim_{j\to\infty} \frac{\log M^{(i_j)}}{n^{(i_j)}}>0,\\
		&\lim_{j\to\infty} w^{(i_j)}=\hat w,\quad \lim_{j\to\infty} \tau^{(i_j)} = \hat \tau,\quad \lim_{j\to\infty} \alpha^{(i_j)}=\hat\alpha .
	\end{align*}
	Lemma~\ref{lem::plotkin bound list decoding} and the positive rate of the encoding schemes yield that the relative $L$-radius of the code used by the $i_j$th encoding scheme at the first stage is at most $\hat w - \hat w^{L+1} + o(1)$ for all $L\ge 1$ as $j\to\infty$. This means that there exist both a codeword of the code used at the first stage and an error pattern with the relative weight $\hat w - \hat w^{L+1} +  o(1)$ such that
	at least $L+1$ codewords are consistent with the output of the channel. To distinguish these $L+1$ possibilities at the second stage, the encoder has to use a code of length $(1-\hat \alpha + o(1))n^{(i_j)}$ and size at least $L+1$ capable of correcting a fraction $(\hat \tau - \hat \alpha (\hat w - \hat w^{L+1} )) / \hat\alpha + o(1)$ of asymmetric errors. However, by Definition~\ref{def::max correctable fraction}, this fraction has to be at most $\tau_{Z}(L+1)$. This implies that
	$$
	\hat \tau\le\hat \alpha(\hat w -\hat w ^{L+1}) + (1-\alpha) \tau_{Z}(L+1),\quad \forall L\ge 1. 
	$$
	By Lemma~\ref{lem::plotkin bound z channel}, $\tau_{Z}(L+1)\to 1/4$ as $L\to \infty$. Since $\hat w^{L+1}\to 0$ as $L\to\infty$, we conclude that $\hat \tau$ also satisfies 
	$$
	\hat \tau \le \hat \alpha \hat w +(1-\hat \alpha)/4.
	$$
	By definition of $\tau_{\max}(w )$ (cf. the system of conditions~\eqref{eq::system}), we conclude that $\hat \tau \le \tau_{\max}(\hat w)\le \tau_{\max}$. This contradiction completes the proof.
    \end{proof}

\section{Conclusion}\label{sec::conclusion}
In this paper, we have discussed two-stage encoding strategies for the Z-channel correcting a fraction of errors. We have proposed an encoding algorithm that uses list-decodable codes on the first stage and high-error low-rate codes on the second stage. This strategy has been shown to correct an optimal fraction of asymmetric errors among all two-stage exponential-sized error-correcting codes.
\section{Acknowledgments}
The authors of this paper are grateful to Amitalok Budkuley, Sidharth Jaggi and Yihan Zhang for the fruitful discussion on the list-decodable codes and for providing a proof for Lemma~\ref{lem::plotkin bound list decoding}.
\bibliographystyle{plain}
\bibliography{ref}
\end{document}

%% file: plotListLow.tex
\begin{figure}[t]
  \centering
\begin{tikzpicture}[thick,scale=0.9]
\pgfplotsset{compat = 1.3}
\begin{axis}[
	legend style={nodes={scale=0.7, transform shape}},
	legend cell align={left},
	width = \columnwidth,
	height = 0.7\columnwidth,
	xlabel = {$\tau$, maximal fraction of errors},
	xlabel style = {nodes={scale=0.8, transform shape}},
	ylabel = {$R_Z(\tau,L)$, asymptotic rate},
	ylabel style={nodes={scale=0.8, transform shape}},
	xmin = 0,
	xmax = 1.0,
	ymin = 0.0,
	ymax = 1.0,
	legend pos = north east]

\addplot[color=green, mark=none,thick] table {1listRateLow.txt};
\addlegendentry{Lower bound on $R_Z(\tau,1)$}
\addplot[color=purple, mark=none,thick] table {2listRateLow.txt};
\addlegendentry{Lower bound on $R_Z(\tau,2)$}
\addplot[color= blue, mark=none,thick] table {3listRateLow.txt};
\addlegendentry{Lower bound on $R_Z(\tau,3)$}
\addplot[color= red, mark=none,thick] table {10listRateLow.txt};
\addlegendentry{Lower bound on $R_Z(\tau,10)$}
\end{axis}
\end{tikzpicture}
  \caption{Lower bounds on the asymptotic rate of binary list-decodable codes for the Z-channel}
  \label{fig::asymptotic rate low bounds list}
\end{figure}
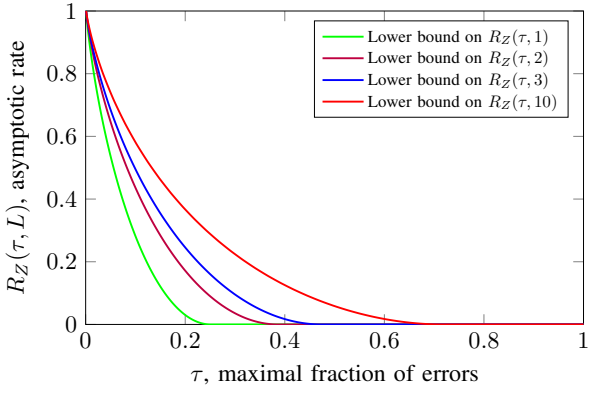

%% file: plot.tex
\begin{figure}[t]
  \centering
\begin{tikzpicture}[thick,scale=0.9]
\pgfplotsset{compat = newest}
\begin{axis}[
    legend pos = {south east},
	legend style={nodes={scale=0.7, transform shape}},
	legend cell align={left},
	width = \columnwidth,
	height = 0.7\columnwidth,
	xlabel = {$\tau$, fraction of errors},
	xlabel style = {nodes={scale=0.8, transform shape}},
	ylabel = {Asymptotic rate},
	ylabel style={nodes={scale=0.8, transform shape}},
	xmin = 0,
	xmax = 1.0,
	ymin = 0.0,
	ymax = 1.0,
	legend pos = north east]

\addplot[color=green, mark=none,thick] table {upBound.txt};
\addlegendentry{One-stage (upper bound),~\cite{mceliece1977new}}
\addplot[color=blue, mark=none,dashed,thick] table {GVBound.txt};
\addlegendentry{One-stage (lower bound),~\cite{gilbert1952comparison,varshamov1957estimate}}
\addplot[color= red, mark=none,dotted,thick] table {newRate.txt};
\addlegendentry{Two-stage (lower bound), Th.~\ref{th: lower bound for two-stage encoding schemes}}
\addplot[color= violet, mark=none,dashdotted,thick] table {low_feedback.tex};
\addlegendentry{Fully adaptive (lower bound), \cite{deppe2020coding}}
 \end{axis}
\end{tikzpicture}
  \caption{Asymptotic rate of error-correcting codes for the Z-channel. Three levels of adaptivity are considered: one-stage (non-adaptive), two-stage and fully adaptive encoding algorithms.}
  \label{fig::asymptotic rate}
\end{figure}
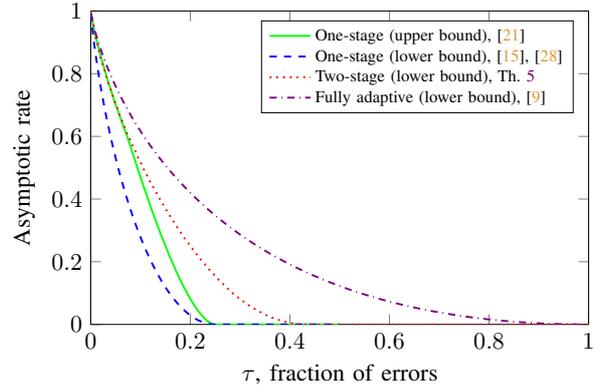